\documentclass[aps,prl,reprint,superscriptaddress,longbibliography]{revtex4-2}

\usepackage[T1]{fontenc}
\usepackage{lmodern}
\usepackage{amsmath,amssymb,bm,mathtools}
\usepackage{amsthm}
\usepackage{graphicx}
\usepackage{booktabs,array}
\usepackage{microtype}
\usepackage[colorlinks=true,linkcolor=blue,citecolor=blue,urlcolor=blue]{hyperref}

\newcommand{\avg}[1]{\left\langle #1\right\rangle}
\newcommand{\fout}{\alpha_{\eta}}
\newcommand{\Dlt}{\Delta\bar C}
\newcommand{\norm}[1]{\left\lVert #1\right\rVert}
\newtheorem{lemma}{Lemma}[section]

\begin{document}

\title{Coverage--information uncertainty for single-helicity light}

\author{Hyoseok Park}
\email{phs137@o.cnu.ac.kr}
\affiliation{Department of Physics, Chungnam National University, Daejeon 34134, Republic of Korea}

\date{August 24, 2026}

\begin{abstract}
A circularly polarized far field must be dark somewhere: its amplitude is a section of a twisted line bundle. We prove this exacts a joint coverage-rate cost that cannot vanish faster than the inverse fourth power of the mode order, and we construct an explicit family of fields attaining that exponent. At lowest order the global optimum is the spin-1 anticoherent state, independently known as optimal for rotations about an unknown axis. The same floor caps photon collection in single-helicity quantum links; the opposite helicity removes it at a cost set by the inverse mode count.
\end{abstract}

\maketitle

A link that shares no orientation with its receiver wants circular polarization, because helicity is the polarization label that rotations preserve. Compact orientation-agnostic links carry it~\cite{Su2022}, nanophotonic structures route it through the spin--orbit coupling of light~\cite{Bliokh2015}, chiral interfaces write an emitter's spin onto it~\cite{Lodahl2017,Sollner2015}, and polarization-encoded quantum channels ride the circular basis between frames that share no alignment~\cite{Laing2010,Kimble2008}. These receivers never see the full field, only the one component their optics select, and the classical limits are pitched at the whole: quasi-isotropy theorems bound how uniform an unrestricted vector field can be~\cite{ChoKwon2025}, and gain, directivity, superdirectivity, stored-energy, and capacity bounds constrain the field as a whole~\cite{Chu1948,Fante1969,GustafssonCapek2019,Ehrenborg2021,Miller2019,KuangMiller2025}. None of them bears on a single helicity. Fixing one changes the problem, and the change is topological.

The local circular basis twists as one moves over the sphere of directions. A helicity-$\sigma$ amplitude is therefore a section of a complex line bundle with first Chern number $-2\sigma$, and a section of a twisted bundle must vanish: every fixed-helicity far field has at least one dark direction~\cite{Palmerduca2024,Palmerduca2025}. Null and polarization-singularity theorems express the same obstruction in classical fields and fix the total index of the far-field singularities to two~\cite{Scott1966,Chen2019,Peng2021,Wen2025,Silveirinha2023}. What these results give is existence. They do not say what the null costs.

The cost is not set by topology but by bandwidth. A dark direction of zero solid angle is harmless; a finite patch of directions below a gain threshold is a coverage outage. A field of order $L$ carries $d_L=L(L+2)\sim L^2$ complex modes, rises no faster than order $L$ from a zero, and climbs back to full brightness only over an angular width $1/L$, paying for the climb elsewhere: shrinking the dark patch forces the pattern to concentrate, and concentration drains the orientation-averaged rate through the concavity of $\log(1+\rho g)$. Coverage and rate are tied together by the forced null, and at wavelength scale the tie is not a small correction: the best possible single-helicity emitter is dark over nearly a fifth of orientations. Below we make the cost exact as a product law over all fields of order $L$, trace it to a dielectric sphere, and follow it into quantum links that read one helicity.

\emph{Setup.---}Let $\bm u\in S^2$ be a propagation direction in a local helicity frame carried by rotations. A unit-power field of helicity $\sigma=\pm1$ and maximum order $L$ is
\begin{equation}
 a(\bm u)=\sum_{\ell=1}^{L}\sum_{m=-\ell}^{\ell}c_{\ell m}\,{}_{\sigma}Y_{\ell m}(\bm u),\qquad \sum_{\ell m}|c_{\ell m}|^2=1,
 \label{eq:field}
\end{equation}
in spin-weighted spherical harmonics~\cite{Goldberg1967}. The amplitude depends on the frame, but its modulus, its zeros, and every quantity below do not; on a frame overlap $a$ picks up only the helicity phase, and the impossibility of fixing that phase globally is the bundle obstruction. With $d\mu=d\Omega/4\pi$ write the normalized gain $g=4\pi|a|^2$, so $\avg{g}=1$. For a threshold $0<\eta<1$ the coverage outage is the fraction of directions below it,
\begin{equation}
 \fout=\avg{\mathbf 1[g\le\eta]},
 \label{eq:outage}
\end{equation}
and for the matched channel $y=\sqrt{\rho g}\,x+n$ at reference signal-to-noise ratio $\rho$ the loss of orientation-averaged coherent capacity, relative to a uniform pattern, is
\begin{equation}
 \Dlt=\log_2(1+\rho)-\avg{\log_2(1+\rho g)}\ge0,
 \label{eq:deficit}
\end{equation}
zero only when the gain is uniform. Under a uniform orientation prior $\fout$ is the probability that the link finds itself below threshold and $\Dlt$ the mean rate it gives up to an isotropic pattern of the same power; but neither refers to a device, and both are functionals of the far-field shape alone. The pair is deliberate, since link budgets work with threshold fractions and information theory with average rates; a bound on the product means that neither quantity can be improved at no cost to the other. We report $\eta=1/2$ and $\rho=10$.

\emph{The bound.---}For fixed $\eta$ and $\rho$ there is a constant $c_{\eta,\rho}>0$, independent of $L$, such that every field of Eq.~\eqref{eq:field} obeys
\begin{equation}
 \boxed{\ \fout\,\Dlt\ \ge\ c_{\eta,\rho}\,L^{-4}.\ }
 \label{eq:bound}
\end{equation}
The peak gain $H=\|g\|_\infty$ controls both sides. A spherical Bernstein inequality~\cite{Arestov1982,Backus1979} caps the slope of a degree-$L$ field at $L\|a\|_\infty$, so the gain leaves its forced zero no faster than order $L^2$; the sub-threshold cap around the zero then has area at least $c_\eta/(HL^2)$, giving
\begin{equation}
 \fout\ \ge\ \frac{c_\eta}{HL^2}.
 \label{eq:capalpha}
\end{equation}
The same slope bound holds at a maximum, where a peak of height $H$ occupies a cap of area of order $L^{-2}$; the strict concavity of the rate then forces
\begin{equation}
 \Dlt\ \ge\ \frac{c_\rho H}{L^2}.
 \label{eq:capdelta}
\end{equation}
The two caps are one statement made twice: a band-limited field cannot vary faster than its highest mode oscillates, so whatever the gain does, it does over at least an angular cell of size $1/L$, at the forced zero and at the peak alike. A superdirective peak shrinks the dark region through Eq.~\eqref{eq:capalpha} but deepens the rate loss through Eq.~\eqref{eq:capdelta} by the same factor of $H$. Multiplying removes the peak and leaves Eq.~\eqref{eq:bound}. Because $d_L\sim L^2$, Eq.~\eqref{eq:bound} is an inverse-square law in the number of one-helicity modes. The full proof, with constants, is in the Supplemental Material below~\cite{SM}.

Separately optimized, the two costs would suggest a smaller product. Eq.~\eqref{eq:capalpha} alone permits the outage to fall as low as order $L^{-4}$ if the field is allowed to grow superdirective, $H\sim L^2$, while Eq.~\eqref{eq:capdelta} alone floors the deficit at order $L^{-2}$ for bounded peak. No single field realizes both extremes at once, since the same $H$ that buys one spends the other. This is the content of Eq.~\eqref{eq:bound}: the coverage gap and the rate deficit cannot be minimized independently, because both are set by the same forced null.

The relation of Eq.~\eqref{eq:bound} to the classical limits is one of kind rather than strength. Chu's bound and its descendants limit how much gain or bandwidth a small volume can support~\cite{Chu1948,Fante1969,GustafssonCapek2019}; modal bounds count the strong channels an aperture can carry~\cite{Miller2019,KuangMiller2025,Ehrenborg2021}; quasi-isotropy theorems limit the flatness of the total power pattern~\cite{ChoKwon2025}. Each of these caps a strength and relaxes as the structure grows richer, whereas Eq.~\eqref{eq:bound} caps how even one helicity can be: it does not care how much power the source spends or how large its aperture is, only how many one-helicity modes it commands, and it alone rests on a topological premise. A receiver that accepts both helicities removes the floor, a point quantified below.

\emph{Sharpness.---}Equation~\eqref{eq:bound} is a lower bound, and a lower bound is only as interesting as its tightness. We test it by computing the achievable frontier directly. Both observables are invariant under rigid rotation, so they depend only on the shape of the coefficient vector, and the achievable pairs $(\fout,\Dlt)$ form a low-dimensional set at each $L$. The deficit is smooth; the outage is a set measure, so we minimize an annealed logistic surrogate and polish the exact fraction, from several dozen random and structured starts per order~\cite{SM}. The evaluation grid matters here: a minimizer scored on a coarse grid hides its sub-threshold cap between the nodes and reports an outage of zero, so every candidate is evaluated on grids that resolve the smallest cap Eq.~\eqref{eq:capalpha} allows and confirmed at doubled resolution. What the search returns are explicit fields, so the frontier is achievable by construction even where global optimality is not certified; bounded below by Eq.~\eqref{eq:bound} and above by these fields, the true frontier is pinned between the two.

The result is the law of Fig.~\ref{fig:bound}(b): the smallest coverage gap and rate deficit found each scale as $L^{-2}$, with fitted exponents $-1.97$ and $-1.96$ over $L=4$ to $16$, so their product follows $L^{-4}$, and the compensated product $\fout\Dlt\,L^4$ stays between $0.08$ and $0.22$ for $L=2$ to $16$. The exponent does not rest on the search. Smoothing the two-zero square-wave profile with a Jackson kernel gives an explicit axisymmetric field at every order, flat away from two simple antipodal zeros, whose outage and deficit are each of order $L^{-2}$ with proven constants, and whose lowest member is the anticoherent state of the next paragraph~\cite{SM}. Pinched between this family and Eq.~\eqref{eq:bound}, the smallest product is $\Theta(L^{-4})$: the exponent four is exact. What the search adds is the prefactor, which the proven constants miss on both sides, the lower bound sitting four orders below the achieved product and the family about two above it~\cite{SM}. The two minima belong to different fields, but the frontier between them is shallow, and neither endpoint is superdirective: the peak gain of the coverage-optimal field grows only from $1.5$ at $L=2$ to $5.1$ at $L=16$, far below the supergain scale $L^2$. The cheap way to empty the sub-threshold set is to spread the pattern flat and keep nothing but the unavoidable null; a field driven superdirective to shrink its dark cap instead sends most of the sphere below threshold.

At $L=1$ there is no frontier to trace. Up to rotation and phase, every spin-1 field is fixed by one parameter, the separation of its two Majorana stars, and this reduction is exact enough to certify a global optimum rather than merely locate one: the anticoherent state~\cite{Zimba2006} minimizes both observables at once for every field in the family, a statement certified with analytic derivative constants rather than by search alone~\cite{SM}. Every other field's gain majorizes it, which extends the rate optimum to every signal-to-noise ratio; a direct threshold scan does the same for the outage~\cite{SM}. This multiplet, the spin-1 state with vanishing angular-momentum expectation, splits the topological charge of the forced null into two simple zeros at antipodal poles and spreads the gain, $\tfrac32\sin^2\theta$, as evenly as a single order allows. These are the states known in quantum metrology as the extremal, least classical points of the spin family~\cite{Bjork2015} and as the optimal detectors of rotations about an unknown axis~\cite{Chryssomalakos2017,Martin2020,Goldberg2018}. The coincidence has a common root: rotation sensing does best when the Majorana stars of the probe sit as far apart on the sphere as they will go, so that no axis is blind, and single-helicity coverage does best when the zeros of the far field are spread the same way, so that no direction is dark. Whether the correspondence persists at higher order, where our optimal fields stay flat but anticoherence comes in degrees, we leave open.

\begin{figure}[t]
\centering
\includegraphics[width=\columnwidth]{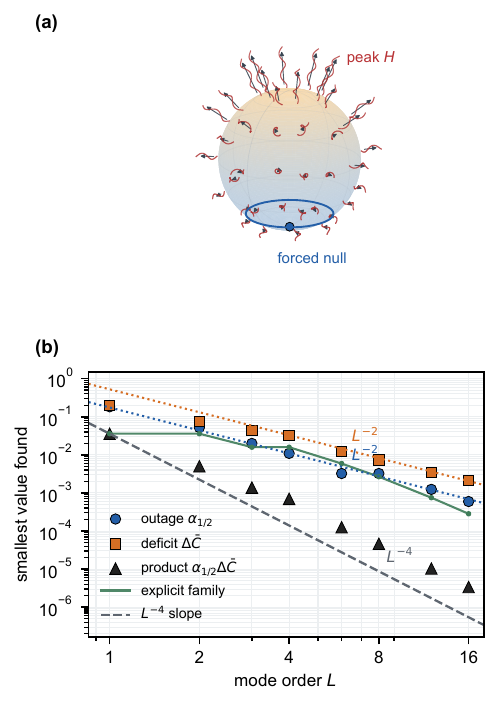}
\caption{The forced null and the coverage--information bound. (a) A single-helicity gain on the momentum sphere: arrows are outgoing rays and the red coils mark their fixed circular polarization. Topology forces a dark direction (blue cap and marker), and finite order $L$ limits how sharply the gain climbs out of it; the peak gain $H$ that narrows the dark cap is the same one that raises the rate deficit. (b) Optimizing over all order-$L$ fields at $\eta=1/2$, $\rho=10$: the smallest outage $\fout$ and deficit $\Dlt$ found each scale as $L^{-2}$ (fitted slopes $-1.97$, $-1.96$), so their product follows the $L^{-4}$ law of Eq.~\eqref{eq:bound}. The thin line is the explicit attaining family, which begins at the $L=1$ optimum, the anticoherent state, and pins the exponent at every order; the dashed line shows the $L^{-4}$ slope anchored below the data (the proven prefactor is smaller~\cite{SM}). The self-dual spheres of Fig.~\ref{fig:realization} lie above this frontier [Fig.~\ref{fig:realization}(c)].}
\label{fig:bound}
\end{figure}

\emph{Realization.---}The low-order single-helicity fields are not abstractions. A self-dual (dual-symmetric) dielectric sphere, whose electric and magnetic Mie coefficients coincide, $a_\ell=b_\ell$, scatters one circular polarization into a helicity-pure field with the topologically forced backward null~\cite{Fu2013,NietoVesperinas2011}; the electromagnetic duality symmetry is what protects the helicity~\cite{Fernandez2013}. A homogeneous sphere cannot satisfy the dual condition at every order, but a radially layered one has enough index freedom to enforce it up to a chosen order and suppress the orders above~\cite{SM}. A crystalline-silicon sphere of radius $122$~nm~\cite{Green2008} approximates the ideal dual dipole, $g(\theta)=\tfrac34(1+\cos\theta)^2$ with $\fout=0.408$ and $\Dlt=0.658$; the full-wave sphere itself gives $0.42$ and $0.69$ at $\rho=10$. Ge/TiO$_2$ and Si/SiN core--shells~\cite{Palik1985} add the quadrupole and octupole, each with opposite-helicity power below $3\times10^{-4}$ [Fig.~\ref{fig:realization}]. Full-wave scattering confirms the null and the multipole content, and both observables, like the helicity purity, are read from the polarization-resolved pattern~\cite{Tischler2014,SM}. These passive spheres sit above the optimal frontier [Fig.~\ref{fig:realization}(c)], and not monotonically closer with order: duality fixes the helicity, not the evenness of the pattern, and the forward-peaked $L=3$ design lands above the $L=2$ one, a gap that traces to the single azimuthal number a plane wave excites at each multipole order~\cite{SM}. Closing that separation is the design problem for a near-isotropic single-helicity source; a reconfigurable aperture of radius $R\simeq L/k$ commands every coefficient and traces the frontier itself.

Because the largest accessible order is set by electrical size, $L\sim ka$~\cite{Chu1948}, the floor $c\,L^{-4}$ is steepest for the most compact sources, which is where an orientation-agnostic circular link operates. The hard cutoff in Eq.~\eqref{eq:field} is not doing hidden work here: a field that is only approximately band limited obeys the same bound with a rescaled constant, because the multipole content of a source of size $ka$ falls off steeply beyond $L\sim ka$ and a small uniform tail only shifts the cap estimates~\cite{SM}. Even an optimal single-helicity emitter of size $ka\approx1$ leaves a fraction $0.18$ of orientations below half its mean gain and loses $0.2$ bits of orientation-averaged rate; the silicon sphere leaves $0.42$. The penalty is an order-one cost on small sources, not an asymptotic correction, and it applies before any specific antenna or scatterer is designed.

\begin{figure}[t]
\centering
\includegraphics[width=\columnwidth]{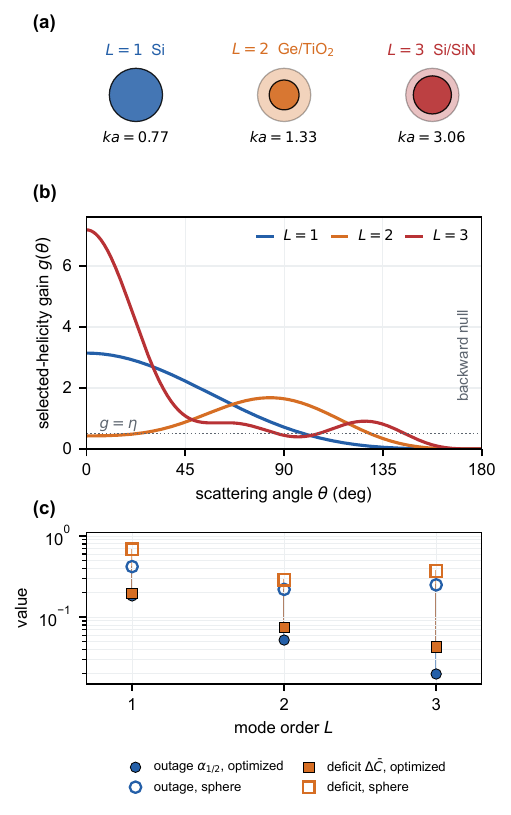}
\caption{(a,b) Self-dual dielectric nanospheres realizing the low-order single-helicity fields: a silicon sphere ($L=1$), a Ge/TiO$_2$ core--shell ($L=2$), and a Si/SiN core--shell ($L=3$). Each is helicity-pure, with opposite-helicity power below $3\times10^{-4}$, and each carries the topologically forced backward null. (c) Each sphere's outage and deficit against the optimized frontier at the same $L$ (Fig.~\ref{fig:bound}(b); numeric values and the reason for the gap in the Supplemental Material~\cite{SM}).}
\label{fig:realization}
\end{figure}

\emph{Quantum links.---}Quantum hardware meets the floor where it is steepest, because the devices that emit one photon at a time into a set helicity are the small-$L$ sources. A chiral interface writes an emitter's spin onto the photon's helicity~\cite{Lodahl2017,Sollner2015}; a polarization qubit rides the circular basis between frames that share no alignment~\cite{Laing2010}; a network of such nodes wants the photon collectable from any direction~\cite{Kimble2008}. Whatever the emitter, the helicity-selected amplitude its receiver reads is itself a section of the twisted bundle, so Eqs.~\eqref{eq:bound}--\eqref{eq:capdelta} apply to the selected channel with its own power as the normalization. For a photon-counting link the gain multiplies the collection efficiency, and every capacity of a lossy channel falls as its transmissivity falls~\cite{Pirandola2017}, so the coverage half carries over unchanged: a single-helicity emitter of order $L$ leaves a solid-angle fraction of at least $c_\eta/(HL^2)$ in which the collected rate drops below threshold, the order-one fraction of the previous paragraph for an interface of atomic scale. The information half does not transfer: a photon-starved rate is linear in the transmissivity, its angular average is fixed by power conservation, and the concavity penalty never comes due. For quantum links the entire price of one helicity is the blind patch, and no directional gain hides it.

\emph{One helicity or two.---}The bound belongs to a single resolved helicity, and impurity tests it in two distinct senses. A receiver that keeps analyzing the $\sigma$ component reads a section of the same twisted bundle at any impurity $\epsilon$, so Eq.~\eqref{eq:bound} holds for it with the selected power as the normalization. A receiver that detects both components responds to the total gain, and there a degree-$L$ evaluation bound, $4\pi|a_-|^2\le\epsilon d_L$, keeps the minority from filling a majority zero while $\epsilon d_L<\eta$; in the other direction, two opposite-helicity beams centered on the polar cores of the explicit family remove total outage once $\epsilon$ reaches $C_\eta\,\eta/d_L$, with $C_\eta$ independent of $L$, and near $3\eta/d_L$ in direct evaluation~\cite{SM}. The least power that removes total outage is therefore
\begin{equation}
 \frac{\eta}{d_L}\le\epsilon_c\le C_\eta\,\frac{\eta}{d_L},\qquad \epsilon_c=\Theta(\eta/d_L).
 \label{eq:crossover}
\end{equation}
This is a change of bundle rather than a repair of the gain. One helicity stays twisted and keeps its zero; the two together form a trivial bundle that may be dark nowhere. The obstruction, and its cost, is exactly what one pays to read a single helicity. The required impurity falls with the inverse mode count, our two sides of Eq.~\eqref{eq:crossover} differing only by an $L$-independent factor. Injecting a calibrated opposite-helicity fraction and locating the impurity at which the polar cores fill would test the scaling directly.

Topology forces a dark direction into every fixed-helicity far field, and finite angular resolution turns that point into a price paid jointly in coverage and in rate. The outage fraction and the coherent-rate deficit cannot both shrink faster than their product allows, $\fout\Dlt\ge c\,L^{-4}$, a law we prove from a Bernstein cap on the forced null, match in exponent with an explicit family of fields, and read off a wavelength-scale dielectric sphere. Reading both helicities lifts it at a power set by the inverse number of modes.

The topology plays a smaller part here than in the quantum-geometric bounds of condensed matter, where a Chern number multiplies the response it constrains~\cite{OnishiFu2024,OnishiFuSF2024,PeottaTorma2015}: the invariant supplies only the existence of the zero, bandwidth sets every scale, and no factor of the Chern number appears in Eq.~\eqref{eq:bound}. It is not settled whether the coverage optimum stays anticoherent beyond $L=1$, where our numerics show only that it stays flat, nor whether a Fisher-information version of Eq.~\eqref{eq:bound} holds for rotation sensing with single-helicity light, where the same constellations serve as probes. What Eq.~\eqref{eq:bound} adds to the classical limits of finite-mode radiation is the cost of fixing the polarization.

\paragraph*{Data availability.} The code and data that reproduce every figure and reported value are available at \url{https://github.com/hyoseokp/single-helicity-coverage-bound}. The raw full-wave field data underlying Table~S3, several gigabytes in size, are available from the author upon reasonable request.

\clearpage
\onecolumngrid
\setcounter{secnumdepth}{3}
\setcounter{equation}{0}
\setcounter{section}{0}
\setcounter{figure}{0}
\setcounter{table}{0}
\renewcommand{\theequation}{S\arabic{equation}}
\renewcommand{\thesection}{S\arabic{section}}
\renewcommand{\thefigure}{S\arabic{figure}}
\renewcommand{\thetable}{S\arabic{table}}

\begin{center}
{\large\bfseries Supplemental Material}
\end{center}
\vspace{4pt}

This material gives a complete proof of the product bound [Eq.~(4)] with explicit constants, and of
its robustness to imperfect band limitation (Sec.~\ref{sec:proof}); the optimization that maps the
achievable frontier and the grid resolution it requires (Sec.~\ref{sec:numerics}); an explicit family
of fields that attains the $L^{-4}$ exponent (Sec.~\ref{sec:family}); the derivation of the
opposite-helicity threshold [Eq.~(7)] (Sec.~\ref{sec:crossover}); and the self-dual sphere designs
with their full-wave check (Sec.~\ref{sec:mie}). Equation, table, and section numbers with the prefix S are internal; unprefixed
numbers refer to the Letter above.

\section{Proof of the product bound}
\label{sec:proof}

\subsection{Setup and the Bernstein slope bound}

Let $a$ be a helicity-$\sigma$ field of maximum order $L$, expanded as in Eq.~(1) with
$\sum_{\ell m}|c_{\ell m}|^2=1$; with the spin-weighted harmonics ${}_\sigma Y_{\ell m}$ orthonormal on
$S^2$ this is $\int_{S^2}|a|^2\,d\Omega=1$. The normalized gain is $g=4\pi|a|^2$, so that
$\avg{g}=(4\pi)^{-1}\int g\,d\Omega=1$; write $H=\norm{g}_\infty$, hence $\norm{a}_\infty=\sqrt{H/4\pi}$.
Because $a$ is a section of a line bundle of Chern number $-2\sigma$ it has at least one zero,
$a(\bm u_0)=0$~\cite{Palmerduca2024,Palmerduca2025}. We use the normalized area measure
$d\mu=d\Omega/4\pi$, for which a geodesic cap of angular radius $r$ has measure
$\mu_{\rm cap}(r)=\tfrac12(1-\cos r)$, and the elementary inequality
\begin{equation}
 \tfrac12(1-\cos r)\ \ge\ \frac{2}{\pi^2}\,r^2,\qquad 0\le r\le\tfrac{\pi}{2},
 \label{eq:capineq}
\end{equation}
which follows because $(1-\cos r)/r^2$ decreases on $[0,\pi/2]$ from $\tfrac12$ to $4/\pi^2$.

\begin{lemma}[Bernstein inequality for spin-weighted fields]\label{lem:bern}
For $a$ as in Eq.~(1), the covariant tangential gradient obeys
\begin{equation}
 \norm{\nabla a}_\infty\ \le\ L\,\norm{a}_\infty .
 \label{eq:bernstein}
\end{equation}
\end{lemma}
\begin{proof}
A unit-speed great circle $\gamma$ is the orbit of the one-parameter rotation subgroup $R_t$ about its
axis, and the same rotations parallel-transport the helicity frame along $\gamma$. Rotations act on
helicity-$\sigma$ fields unitarily and multiplet by multiplet, through the Wigner matrices
$D^\ell(R_t)$, whose entries are trigonometric polynomials of degree $\ell$ in $t$, because the
generator of $R_t$ has eigenvalues $-\ell,\dots,\ell$ in the spin-$\ell$ representation. In the
transported frame the restriction $t\mapsto a(\gamma(t))$ is therefore a trigonometric polynomial of
degree at most $L$, and the Bernstein inequality on the circle bounds its derivative by
$L\,\sup_t|a(\gamma(t))|\le L\norm{a}_\infty$. The covariant derivative of $a$ along $\gamma$ is
exactly this $t$ derivative, and every tangent direction is tangent to a great circle; maximizing over
directions gives \eqref{eq:bernstein}. See Refs.~\cite{Arestov1982,Backus1979} for the scalar case;
the constant here is exactly $L$, with no correction from the spin weight.
\end{proof}

Since $|a|$ is frame independent, integrating \eqref{eq:bernstein} along the minimizing geodesic from
any base point $\bm u_b$ gives, for every $\bm u$ at geodesic distance $\theta$ from $\bm u_b$,
\begin{equation}
 \bigl|\,|a(\bm u)|-|a(\bm u_b)|\,\bigr|\ \le\ L\,\norm{a}_\infty\,\theta .
 \label{eq:lip}
\end{equation}

\subsection{Coverage cap}

\begin{lemma}\label{lem:alpha}
$\displaystyle \fout\ \ge\ \frac{2\eta}{\pi^2}\,\frac{1}{H L^2}.$
\end{lemma}
\begin{proof}
Take $\bm u_b=\bm u_0$ in \eqref{eq:lip}. Since $a(\bm u_0)=0$,
$|a(\bm u)|\le L\sqrt{H/4\pi}\,\theta$, so $g(\bm u)=4\pi|a(\bm u)|^2\le L^2H\theta^2$. Thus $g\le\eta$
throughout the cap $\theta\le r_\alpha\equiv\sqrt{\eta/H}\,/L$, which for $\eta<1\le H$ satisfies
$r_\alpha\le\pi/2$. By \eqref{eq:capineq},
$\fout\ge\mu_{\rm cap}(r_\alpha)\ge (2/\pi^2)\,r_\alpha^2=(2\eta/\pi^2)/(HL^2)$.
\end{proof}

\subsection{Rate deficit}

The deficit is a nonnegative Jensen residual,
\begin{equation}
 \Dlt=\avg{\Phi_\rho(g)},\qquad
 \Phi_\rho(x)=\log_2(1{+}\rho)+\frac{\rho(x-1)}{(1+\rho)\ln2}-\log_2(1+\rho x),
 \label{eq:phi}
\end{equation}
the gap between $\log_2(1+\rho x)$ and its tangent at $x=1$. Since $\log_2(1+\rho x)$ is strictly
concave, $\Phi_\rho\ge0$, with $\Phi_\rho(1)=0$; and
$\Phi_\rho'(x)=\frac{\rho}{\ln2}\!\left[\frac{1}{1+\rho}-\frac{1}{1+\rho x}\right]$ is negative for
$x<1$ and positive for $x>1$, so $\Phi_\rho$ decreases on $[0,1]$ and increases on $[1,\infty)$.

\begin{lemma}\label{lem:delta}
There are constants $H_0(\rho)\ge4$ and $\kappa_\rho>0$ (at $\rho=10$, $H_0\approx23$ and $\kappa_\rho\approx0.16$) such that
$\Dlt\ge \dfrac{\kappa_\rho}{2\pi^2}\,\dfrac{H}{L^2}$ whenever $H\ge H_0$, and
$\Dlt\ge \dfrac{\Phi_\rho(1/4)}{2\pi^2 H_0}\,\dfrac{1}{L^2}$ whenever $H<H_0$.
\end{lemma}
\begin{proof}
\emph{Large peak, $H\ge H_0$.} Let $\bm u_\ast$ be a maximum, $|a(\bm u_\ast)|=\sqrt{H/4\pi}$. By
\eqref{eq:lip} with $\bm u_b=\bm u_\ast$, $|a(\bm u)|\ge\sqrt{H/4\pi}\,(1-L\theta)\ge\tfrac12\sqrt{H/4\pi}$
for $\theta\le r_\beta\equiv1/(2L)$, hence $g\ge H/4$ on that cap, whose measure is at least
$(2/\pi^2)r_\beta^2=1/(2\pi^2L^2)$ by \eqref{eq:capineq}. Restricting the average \eqref{eq:phi} to the
cap and using monotonicity, $\Dlt\ge (2\pi^2L^2)^{-1}\Phi_\rho(H/4)$. For large argument
$\Phi_\rho(x)=\frac{\rho}{(1+\rho)\ln2}\,x-\log_2(\rho x)+O(1)$, so the linear term dominates: there is
$H_0\ge4$ with $\Phi_\rho(H/4)\ge\kappa_\rho H$ for $H\ge H_0$, where
$\kappa_\rho=\rho/[8(1+\rho)\ln2]$. This gives the first bound.

\emph{Bounded peak, $H<H_0$.} Reusing $g\le L^2H\theta^2\le L^2H_0\theta^2$ near $\bm u_0$, we have
$g\le1/4$ on the cap $\theta\le1/(2L\sqrt{H_0})$, of measure at least
$1/(2\pi^2H_0L^2)$. On $[0,1/4]$ the residual is bounded below by its value at the right endpoint,
$\Phi_\rho(g)\ge\Phi_\rho(1/4)>0$, so $\Dlt\ge \Phi_\rho(1/4)/(2\pi^2H_0L^2)$.
\end{proof}

\subsection{Product}

Multiplying Lemmas~\ref{lem:alpha} and \ref{lem:delta} cancels the peak gain. For $H\ge H_0$,
\begin{equation}
 \fout\,\Dlt\ \ge\ \frac{2\eta}{\pi^2HL^2}\cdot\frac{\kappa_\rho H}{2\pi^2L^2}
 \ =\ \frac{\eta\,\kappa_\rho}{\pi^4}\,L^{-4}.
 \label{eq:prod1}
\end{equation}
For $H<H_0$, Lemma~\ref{lem:alpha} gives $\fout\ge(2\eta/\pi^2)/(H_0L^2)$ and Lemma~\ref{lem:delta} gives
$\Dlt\ge\Phi_\rho(1/4)/(2\pi^2H_0L^2)$, so
\begin{equation}
 \fout\,\Dlt\ \ge\ \frac{\eta\,\Phi_\rho(1/4)}{\pi^4H_0^{\,2}}\,L^{-4}.
 \label{eq:prod2}
\end{equation}
Taking the smaller of \eqref{eq:prod1} and \eqref{eq:prod2} proves Eq.~(4) with
\begin{equation}
 c_{\eta,\rho}=\frac{\eta}{\pi^4}\,\min\!\Big\{\kappa_\rho,\ \Phi_\rho(1/4)/H_0^{\,2}\Big\}>0.
\end{equation}
At $\eta=1/2$, $\rho=10$ the numbers are $\kappa_\rho=0.164$, $\Phi_\rho(1/4)=0.668$, and $H_0\approx23$
(the smallest $H$ with $\Phi_\rho(H/4)\ge\kappa_\rho H$), giving $c_{\eta,\rho}\approx7\times10^{-6}$.
The compensated product achieved in Table~\ref{tab:opt} is of order $10^{-1}$, and the explicit family
of Sec.~\ref{sec:family} caps the minimum from above at the same power of $L$: the exponent of Eq.~(4)
is exact, while the constant, assembled from worst-case cap estimates, is conservative by four orders
of magnitude and is not claimed optimal.
The estimate uses $a$ only through $|a|$, so every step is independent of the local helicity frame, and
topology enters once, to place the zero of Lemma~\ref{lem:alpha}. The same argument applied to a zero of
order $m$ replaces $\eta$ by $\eta^{1/m}$ in Lemma~\ref{lem:alpha}; since the generic forced zero is
simple this only improves the constant, and it never introduces a factor of the Chern number, because the
Poincar\'e--Hopf index is carried by a single coalesced zero as often as by separated ones.

\subsection{Robustness to imperfect band limitation}
\label{sec:tail}

A physical source is never strictly band limited, so we record how Eq.~(4) degrades under a small
high-order tail. Write $\tilde a=a+a_t$, with $a$ the modes of $\ell\le L$ and $a_t$ the rest, the total
field power normalized, $\tilde g=4\pi|\tilde a|^2$, $\tilde H=\norm{\tilde g}_\infty\ge1$, and suppose
the tail is uniformly small,
\begin{equation}
 \tau\ \equiv\ 4\pi\norm{a_t}_\infty^2\ \le\ \min\!\big(\eta/16,\ 1/64\big).
 \label{eq:tailsize}
\end{equation}
The topological zero belongs to the full field, $\tilde a(\bm u_0)=0$, whatever its bandwidth, so
$|a(\bm u_0)|\le\norm{a_t}_\infty$, and \eqref{eq:lip} applied to the band-limited part gives
$\sqrt{\tilde g(\bm u)}\le2\sqrt\tau+2L\sqrt{\tilde H}\,\theta$ at distance $\theta$ from the zero. With
\eqref{eq:tailsize} the sub-threshold radius of Lemma~\ref{lem:alpha} shrinks by at most a factor of
$4$, so the coverage cap survives with $\eta$ replaced by $\eta/16$. At a maximum the same displacement
argument gives $\tilde g\ge\tilde H/4$ on a cap of radius $1/(8L)$ rather than $1/(2L)$, and in the
bounded-peak branch $\tilde g\le\tfrac14$ on a cap smaller by the same factor of $4$ in radius. Each of
the three cap measures is reduced by at most $16$, and the product bound follows as in
Eqs.~\eqref{eq:prod1}--\eqref{eq:prod2}:
\begin{equation}
 \fout(\tilde a)\,\Dlt(\tilde a)\ \ge\ \frac{c_{\eta,\rho}}{256}\,L^{-4},
 \label{eq:tailbound}
\end{equation}
with the same exponent and the constant of Eq.~(4) degraded by the stated factor. The uniform tail of a
real source is small in precisely this sense: for a current of electrical size $ka$ the multipole
amplitudes fall off faster than exponentially beyond $\ell\sim ka$, so \eqref{eq:tailsize} is met once
$L$ exceeds $ka$ by a few units. The hard cutoff in Eq.~(1) is a normalization of the mode count, not a
physical assumption.

\section{Optimization and its convergence}
\label{sec:numerics}

\subsection{Method}
Both $\fout$ and $\Dlt$ are invariant under rigid rotation, so they are functions of the
$SU(2)$-invariant shape of $\{c_{\ell m}\}$, and the achievable region is low dimensional at small $L$.
We minimize $\fout$, $\Dlt$, and the scalarizations $\fout+\lambda\Dlt$ over the unit sphere
$\sum|c_{\ell m}|^2=1$. The deficit is smooth with an analytic gradient. The outage is not, so we minimize
a logistic surrogate $\avg{[1+e^{(g-\eta)/\tau}]^{-1}}$ and anneal $\tau\downarrow0$, then polish by
coordinate descent on the exact fraction. Each run starts from a random point on the sphere or from a
structured seed (a coherent spike, the anticoherent $\ell{=}1$ state, or a superposition of antipodal
spikes); we keep the best of several dozen starts per $L$.

\subsection{Resolution}
The outage is a set measure, so the evaluation grid must resolve the sub-threshold cap. A field with peak
gain $H$ has, by Lemma~\ref{lem:alpha}, a cap of angular radius $\gtrsim\sqrt{\eta/H}/L$; a grid coarser
than this reports a spuriously small outage, and a minimizer run against a coarse grid exploits the gap by
hiding the cap between nodes. We therefore evaluate every reported field on a Gauss--Legendre grid in
$\cos\theta$ with $\ge 80L$ nodes and a uniform $\phi$ grid with $\ge 16L$ nodes, and confirm each value
against a doubled grid. Table~\ref{tab:opt} lists the converged minima at $\eta=1/2$, $\rho=10$. These are
\emph{achievable} (upper-bound) frontiers: every entry is an explicit field, so the bound side is
automatic, but for $L\ge2$ global optimality is not certified. The exception is $L=1$, treated next, where
the coefficient space is small enough to certify the optimum outright.

\begin{table}[h]
\centering\small
\begin{tabular}{@{}rccccc@{}}
\toprule
$L$ & $\fout^{\min}$ & $\fout^{\min}L^2$ & $\Dlt^{\min}$ & $H$ at min & $\fout\Dlt\,L^4$\\
\midrule
2  & $5.2\times10^{-2}$ & 0.21 & $7.5\times10^{-2}$ & 1.5 & 0.08\\
3  & $2.0\times10^{-2}$ & 0.18 & $4.5\times10^{-2}$ & 2.1 & 0.11\\
4  & $1.1\times10^{-2}$ & 0.18 & $3.3\times10^{-2}$ & 2.0 & 0.18\\
6  & $3.2\times10^{-3}$ & 0.12 & $2.0\times10^{-2}$ & 3.4 & 0.16\\
8  & $3.2\times10^{-3}$ & 0.21 & $1.2\times10^{-2}$ & 3.9 & 0.19\\
12 & $1.2\times10^{-3}$ & 0.18 & $6.0\times10^{-3}$ & 4.4 & 0.22\\
16 & $6.0\times10^{-4}$ & 0.15 & $4.0\times10^{-3}$ & 5.1 & 0.22\\
\bottomrule
\end{tabular}
\caption{Optimal single-helicity frontier at $\eta=1/2$, $\rho=10$, from converged multi-start
optimization. Fitted exponents over $L=4$--$16$ are $-1.97$ ($\fout$) and $-1.96$ ($\Dlt$); the
compensated product $\fout\Dlt L^4$ stays between $0.08$ and $0.22$.}
\label{tab:opt}
\end{table}

The minimizer of the outage is not superdirective: its peak gain rises only from $1.5$ at $L=2$ to $5.1$
at $L=16$, far below the supergain scale $L^2$. Consistent with Lemma~\ref{lem:delta}, emptying the
sub-threshold set is done by flattening the pattern and leaving only the unavoidable null, not by spiking
it; a field driven superdirective in an attempt to shrink the cap instead sends most of the sphere below
threshold and raises the outage.

\subsection{The $L=1$ case}
At $L=1$ the coefficient space is a single spin-1 multiplet, and up to rigid rotation and phase every
shape belongs to a one-parameter family labeled by the separation of its two Majorana stars. This
reduction is exact, so the optimum over it can be certified rather than merely searched for. A canonical
representative of the family is
\begin{equation}
 |\psi(t)\rangle=\cos t\,|1,1\rangle+\sin t\,|1,{-1}\rangle,\qquad t\in[0,\pi/4],
 \label{eq:l1family}
\end{equation}
with $t=0$ the coherent (dual-dipole) state and $t=\pi/4$ the anticoherent state~\cite{Zimba2006}; every
$L=1$ field is a rotation of some $|\psi(t)\rangle$. With $c=\cos\theta$, $A=\cos t\,(1-c)/2$, and
$B=\sin t\,(1+c)/2$, the gain is $g=3\big[A^2+B^2+2AB\cos2\phi\big]$, and the $\phi$ integral in both
$\fout(t)$ and $\Dlt(t)$ is elementary, leaving a single quadrature in $c$.

The slice makes the certificate elementary, because the $t$ derivative of the gain has a closed form,
$\partial_t g=3\big[c\sin2t+\tfrac12(1-c^2)\cos2t\cos2\phi\big]$, bounded pointwise by
$|\partial_t g|\le3$. A scan of $20001$ points, spacing $h=3.9\times10^{-5}$, shows both observables
decreasing monotonically to $t=\pi/4$ and certifies the minimum with analytic constants. For the
outage: $g_t\le g_{t_i}+3h/2$ pointwise at the nearest scan point, so the sublevel set of $g_t$ at
$\eta$ contains that of $g_{t_i}$ at $\eta-3h/2$, and evaluating the scan at the shifted threshold
proves $\fout\ge0.18348$ for every $L=1$ field, against the anticoherent value $0.18350$: the global
minimizer to within $2.4\times10^{-5}$. For the deficit:
$|\partial_t\Dlt|\le(\rho/\ln2)\avg{|\partial_t g|}\le24$ at $\rho=10$, which proves
$\Dlt\ge0.19528$ for every field against the anticoherent $0.19574$, global to within
$5\times10^{-4}$; the margin is set by the deliberately crude derivative constant, the observed slope
never exceeding $0.97$. An independent brute-force evaluation on random spin-1 states, located on the
slice through their own Majorana stars, agrees to within $10^{-5}$. The optimum is stronger than a
minimum at one threshold and one signal-to-noise ratio. The gain of every other member of the family majorizes the
anticoherent gain (its Lorenz curve lies on or below all others, to within numerical noise, over the full
range of $t$), so the anticoherent state minimizes the average of every convex function of the gain at
once; the deficit is such an average [Eq.~\eqref{eq:phi}], so $\Dlt$ is minimized for every $\rho$. The
outage is not a convex functional and we verify it directly: scans at $\eta=0.1$ to $0.9$ all return the
anticoherent state as the minimizer.

The anticoherent multiplet, the spin-1 state with $\avg{\bm J}=0$, splits the topological index~$2$ of the
forced null into two simple zeros at antipodal poles, rather than the single double zero of the coherent
state, and in the frame aligned with its own symmetry axis its gain is $g=\tfrac32\sin^2\theta$. At
$\eta=1/2$ the observables are closed form,
\begin{equation}
 \fout=1-\sqrt{2/3}\simeq0.1835,\qquad
 \Dlt=\log_2(1+\rho)+\frac{2}{\ln 2}\Big[1-\frac{\mathrm{artanh}\sqrt a}{\sqrt a}\Big]
 \simeq0.195,
\end{equation}
with $a=\tfrac32\rho/(1+\tfrac32\rho)=15/16$ at $\rho=10$; the corresponding dual-dipole values are
$0.408$ and $0.658$. These are the numbers quoted for the certified $ka\approx1$ optimum in the Letter.

\section{An explicit family attaining the exponent}
\label{sec:family}

The optimization of Sec.~\ref{sec:numerics} finds fields whose product follows $L^{-4}$, but a search
cannot by itself pin an exponent. This section removes the dependence on the search: an explicit,
closed-form family of admissible fields obeys $\fout\,\Dlt\le C_{\eta,\rho}L^{-4}$, so together with
Eq.~(4) the smallest product over degree-$L$ fields is $\Theta(L^{-4})$ and the exponent four is exact.

\subsection{Construction}

Let $J_n(u)=\gamma_n\big[\sin(nu/2)/\sin(u/2)\big]^4$ be the Jackson kernel on the circle,
$\gamma_n=3/(2n^3+n)$ fixing $(2\pi)^{-1}\!\int_{-\pi}^{\pi}J_n=1$, and let
$q(\theta)=\mathrm{sgn}(\sin\theta)$ be the square wave. For maximum order $L\ge1$ set
$n=\lfloor(L+3)/2\rfloor$ and take the smoothed profile
\begin{equation}
 b_L\ =\ q\ast J_n,
 \label{eq:family}
\end{equation}
the circular convolution. Because $\hat q$ is supported on odd harmonics and $\hat J_n$ on
$|k|\le2n-2$, the profile $b_L$ is an odd sine polynomial of degree at most $2n-3\le L$; through
$\sin(o\theta)=\sin\theta\;U_{o-1}(\cos\theta)$, with $U_{o-1}$ the Chebyshev polynomial, it lies in
$\sin\theta\cdot\{\text{polynomials of degree}\le L-1\text{ in }\cos\theta\}$, which is exactly the
span of the axisymmetric spin-weighted harmonics ${}_\sigma Y_{\ell0}\propto\sin\theta\,
P_\ell'(\cos\theta)$ with $1\le\ell\le L$. So $b_L$ is an admissible field of Eq.~(1). It vanishes at
the two poles, where the topological index sits as two antipodal zeros, and at $L=1$ it reduces to
$b_1\propto\sin\theta$: the family begins at the anticoherent optimum of Sec.~\ref{sec:numerics} and
extends it to every order.

\subsection{Accuracy, outage, and deficit}

Write $d(\theta)=\min(\theta,\pi-\theta)$ for the distance to the nearest zero. Since
$\sin(u/2)\ge u/\pi$ on $(0,\pi]$, the kernel obeys $J_n(u)\le(3\pi^4/2)\,n^{-3}u^{-4}$, and
integrating the envelope gives the mass bound
$(2\pi)^{-1}\!\int_{|u|\ge d}J_n\le(\pi^3/2)(nd)^{-3}$. The convolution error is twice the kernel mass
on the arc $(\theta,\theta+\pi)$, whose distance from the origin is $d(\theta)$, so
\begin{equation}
 0\ \le\ 1-b_L(\theta)\ \le\ \min\!\big\{2,\ \pi^3(n\,d(\theta))^{-3}\big\},\qquad b_L\le1 .
 \label{eq:famacc}
\end{equation}
Integrating \eqref{eq:famacc} over the sphere, $\avg{1-b_L}\le2\pi^2/n^2$, hence the mean
$M=\avg{b_L^2}\ge1-4\pi^2/n^2$. Three consequences follow for the normalized gain $g=b_L^2/M$.

\emph{Outage.} If $d(\theta)>d_*\equiv(\pi/n)(1-\sqrt\eta)^{-1/3}$ then
$1-b_L<1-\sqrt\eta$, so $g>\eta$; the sub-threshold set is contained in the two polar caps of radius
$d_*$, and
\begin{equation}
 \fout\ \le\ 1-\cos d_*\ \le\ \frac{\pi^2}{2}\,(1-\sqrt\eta)^{-2/3}\,n^{-2}.
 \label{eq:famalpha}
\end{equation}

\emph{Deficit.} The pointwise bound
$\log_2\!\big[(1{+}\rho)/(1{+}\rho x)\big]\le\rho\,(1-x)_+/\ln2$ and
$1-g\le(1-b_L^2)/M\le2(1-b_L)/M$ give
\begin{equation}
 \Dlt\ \le\ \frac{\rho}{\ln2}\,\avg{(1-g)_+}\ \le\ \frac{4\pi^2\rho}{M\ln2}\,n^{-2}.
 \label{eq:famdelta}
\end{equation}

\emph{Product.} With $n\ge(L+2)/2$, multiplying \eqref{eq:famalpha} and \eqref{eq:famdelta} yields,
for every $L\ge15$ (so that $M\ge\tfrac12$),
\begin{equation}
 \fout\,\Dlt\ \le\ \frac{32\pi^4\rho\,(1-\sqrt\eta)^{-2/3}}{M\ln2}\,L^{-4}
 \ \le\ 2\times10^{5}\,L^{-4}\quad(\eta=\tfrac12,\ \rho=10).
 \label{eq:fambound}
\end{equation}
The constants are conservative by design; the point of \eqref{eq:fambound} is the power. Together with
Eq.~(4), $c_{\eta,\rho}L^{-4}\le\min\fout\Dlt\le C_{\eta,\rho}L^{-4}$: the exponent is exact, for
every threshold and signal-to-noise ratio, by construction rather than by search.

\begin{table}[h]
\centering\small
\begin{tabular}{@{}rcccc@{}}
\toprule
$L$ & $\fout$ & $\Dlt$ & $\fout\Dlt\,L^4$ & $H$\\
\midrule
1  & 0.1835   & 0.1957   & 0.036 & 1.50\\
4  & 0.1260   & 0.1260   & 4.1   & 1.30\\
8  & 0.0541   & 0.0494   & 10.9  & 1.09\\
16 & 0.0180   & 0.0158   & 18.6  & 1.03\\
32 & 0.00516  & 0.00447  & 24.2  & 1.008\\
64 & 0.00138  & 0.00119  & 27.6  & 1.002\\
96 & 0.000627 & 0.000540 & 28.8  & 1.001\\
\bottomrule
\end{tabular}
\caption{The attaining family \eqref{eq:family} at $\eta=1/2$, $\rho=10$. The scaled observables settle
to $\fout\,n^2\approx1.50$ and $\Dlt\,n^2\approx1.30$ with $n=\lfloor(L+3)/2\rfloor$, moving by under
one percent from $L=32$ on, so the compensated product approaches $16\times1.95\approx31$: far above
the optimized frontier of Table~\ref{tab:opt}, but the same power of $L$, with no logarithmic drift.
The peak gain $H$ falls to $1$: the family attains the exponent while becoming asymptotically
isotropic.}
\label{tab:family}
\end{table}

Table~\ref{tab:family} lists the family's observables. Two remarks. First, the family is about a
factor of $10^2$ above the optimized frontier in the product, so it settles the exponent but not the
prefactor; the frontier fields of Table~\ref{tab:opt} remain the best known. Second, the profile
realizes the structure the optimizer discovers on its own: two well-separated simple zeros and a gain
that is flat everywhere else.

\section{Opposite-helicity threshold}
\label{sec:crossover}

Write the field with both helicities, powers $1-\epsilon$ in the wanted helicity $\sigma$ and
$\epsilon$ in the minority $-\sigma$. Two receivers must be kept apart. One keeps analyzing the
$\sigma$ component alone: what it reads is the majority section, a section of the same twisted bundle
at every $\epsilon$, so its outage and deficit, normalized by the selected power, obey Eq.~(4) no
matter how impure the field is. The other detects both helicities; the two circular components are
orthogonal polarizations in every direction, so their powers add with no interference term,
$g_{\rm tot}=g_{\sigma}+g_{-\sigma}$. The threshold $\epsilon_c$ concerns this total-power receiver.

The lower side is universal. The minority part $a_-$ is a degree-$L$ spin-$(-\sigma)$ field with
$\int|a_-|^2\,d\Omega=\epsilon$, i.e.\ coefficients $b_{\ell m}$ with
$\sum|b_{\ell m}|^2=\epsilon$. By Cauchy--Schwarz on the pointwise expansion,
\begin{equation}
 4\pi|a_-(\bm u)|^2\le 4\pi\Big(\sum|b_{\ell m}|^2\Big)\sum_{\ell m}|{}_{-\sigma}Y_{\ell m}(\bm u)|^2
 =\epsilon\,d_L,\qquad d_L=L(L+2),
 \label{eq:evalbound}
\end{equation}
where the addition theorem gives $\sum_{\ell m}|{}_{-\sigma}Y_{\ell m}|^2=d_L/4\pi$. At a zero of the
majority field the total gain is the minority gain alone, at most $\epsilon d_L$, which sits below
threshold while $\epsilon d_L<\eta$; the outage stays positive there for every majority field. This is
the lower side, $\epsilon_c\ge\eta/d_L$.

The upper side is constructive, and we prove it for the attaining family of Sec.~\ref{sec:family},
whose gain falls below threshold only on two polar caps of radius $d_*=O(1/L)$
[Eqs.~\eqref{eq:famacc}--\eqref{eq:famalpha}, with the threshold read at $\eta/(1-\epsilon)$]. Give the
minority two bumps, one per pole,
\begin{equation}
 \psi_N(\theta)=\frac{1+\cos\theta}{2}\,K_{n'}(\theta),\qquad
 K_{n'}(\theta)=\Big[\frac{\sin(n'\theta/2)}{n'\sin(\theta/2)}\Big]^2,
 \label{eq:bump}
\end{equation}
and its mirror image $\psi_S$ at the south pole, carrying distinct azimuthal indices. Each bump is
admissible: $(1+\cos\theta)/2$ times an even trigonometric polynomial of degree $n'-1\le L-1$ lies in
the degree-$L$ minority space, is regular and maximal at its own pole, and carries the minority
bundle's forced double zero at the opposite one. Distinct azimuthal indices make the worst relative
phase explicit, $g_{\rm tot}\ge(1-\epsilon)g_M+\big(|\psi_N|-|\psi_S|\big)^2$ pointwise in the
appropriate normalization. Choosing $n'\simeq\pi/d_*$ puts the whole majority cap inside the first
Fej\'er lobe, where $K_{n'}\ge(2/\pi)^2$ and the far bump contributes only at relative order
$n'^{-2}$; the norm bound $\avg{\psi_N^2}\le\pi^2/(2n'^2)$, obtained exactly as in
Sec.~\ref{sec:family}, then floors the bump's per-unit-power gain at $2n'^2/\pi^2$ times the stated
cap constants. Collecting the factors, a minority power $\epsilon=C_\eta\,\eta/d_L$ removes total
outage, with $C_\eta$ explicit, of order $10^3$ at $\eta=1/2$, valid once $\epsilon\le1/10$; the chain
is deliberately loose. Evaluated directly, the same construction removes total outage already at
$\epsilon$ between $3.1$ and $3.5$ times $\eta/d_L$ for $L=2$ to $16$, with the kernel order chosen
numerically. Hence $\eta/d_L\le\epsilon_c\le C_\eta\,\eta/d_L$ [Eq.~(7)]: $\epsilon_c=\Theta(\eta/d_L)$,
the two sides differing by a factor independent of $L$, and neither the constant of the chain nor the
numerical $3.5$ is claimed sharp. The two helicities together form a trivial bundle, which is why their
sum may be nowhere dark; the escape is a change of bundle rather than a repair of the scalar gain.

\section{Self-dual dielectric spheres}
\label{sec:mie}

A sphere whose electric and magnetic Mie coefficients coincide, $a_\ell=b_\ell$ for $1\le\ell\le L$,
satisfies the dual-symmetry (generalized Kerker) condition and preserves helicity: illuminated by one
circular polarization it scatters into the same helicity, with no cross-helicity
component~\cite{Fu2013,NietoVesperinas2011}. The scattered far field is then a pure helicity-$\sigma$
section of order $L$, with the topologically forced null in the backward direction. For a homogeneous
sphere $a_\ell=b_\ell$ cannot hold at every $\ell$, but a radially layered sphere has enough freedom to
enforce it up to a chosen order; we solve the Aden--Kerker coated-Mie equations for the layer indices and
radii that set $a_\ell=b_\ell$ for $\ell\le L$ and suppress $\ell>L$. Table~\ref{tab:mie} lists the three
designs, with electrical size $ka$, the fraction $P_{\rm opp}$ of scattered power in the opposite
helicity, and the resulting observables. The dual dipole ($L=1$) is closed form,
$g(\theta)=\tfrac34(1+\cos\theta)^2$, with a second-order backward zero.

\begin{table}[h]
\centering\small
\begin{tabular}{@{}clcccc@{}}
\toprule
$L$ & material & $ka$ & $P_{\rm opp}$ & $\fout$ & $\Dlt(\rho{=}10)$\\
\midrule
1 & Si (homogeneous)       & 0.77 & $1.7\times10^{-4}$ & 0.42 & 0.69\\
2 & Ge/TiO$_2$ core--shell & 1.33 & $5.2\times10^{-5}$ & 0.22 & 0.29\\
3 & Si/SiN core--shell     & 3.06 & $2.9\times10^{-4}$ & 0.25 & 0.37\\
\bottomrule
\end{tabular}
\caption{Self-dual dielectric spheres realizing the low-order single-helicity fields; $P_{\rm opp}$ is the
fraction of scattered power in the opposite helicity.}
\label{tab:mie}
\end{table}

We verify each design with a full-wave finite-difference time-domain scattering simulation under
circularly polarized plane-wave illumination, projecting the far field onto the two helicity components.
The simulations reproduce the backward null and the target multipole content, with residual
opposite-helicity power at the level quoted in Table~\ref{tab:mie} and higher-order ($\ell>L$) leakage
below $2\times10^{-3}$. These passive spheres lie above the optimal frontier of Table~\ref{tab:opt}, and
not monotonically closer with order (Table~\ref{tab:sphere_vs_frontier}): a Kerker scatterer is
helicity-pure but not coverage-optimal.

\begin{table}[h]
\centering\small
\begin{tabular}{@{}rcccc@{}}
\toprule
$L$ & $\fout$ (sphere) & $\fout$ (frontier) & $\Dlt$ (sphere) & $\Dlt$ (frontier)\\
\midrule
1 & 0.42 & 0.18   & 0.69 & 0.20\\
2 & 0.22 & 0.052  & 0.29 & 0.075\\
3 & 0.25 & 0.020  & 0.37 & 0.045\\
\bottomrule
\end{tabular}
\caption{Each self-dual sphere's outage and deficit against the optimal frontier at the same $L$
(Table~\ref{tab:opt} for $L\ge2$; the closed-form $L=1$ optimum of Sec.~\ref{sec:numerics}). The gap
widens with $L$, roughly twelvefold in outage by $L=3$.}
\label{tab:sphere_vs_frontier}
\end{table}

The shortfall is not a materials limitation but a selection rule. A plane wave of definite helicity,
expanded in vector spherical harmonics, populates only one azimuthal number per order, and an isotropic
scatterer conserves it, so the sphere's field is confined to an $L$-dimensional slice of the full
$d_L$-dimensional coefficient space, one coefficient per order. The optimum needs more than that slice:
the certified $L=1$ optimum [$t=\pi/4$ in Eq.~\eqref{eq:l1family}] is an equal superposition of the two
$m=\pm1$ components of the multiplet, a mixing no single plane wave scattered from an isotropic particle
can produce. The sphere accordingly sits at the coherent endpoint of the family, $t=0$ exactly, rather
than partway toward the anticoherent optimum. Reaching the frontier calls for a different platform:
multiple coherent beams, a shaped (non-spherical) scatterer, or a reconfigurable aperture of radius
$R\simeq L/k$, which commands every coefficient and traces the frontier itself. The code, designs, and
postprocessed simulation files are available at
\url{https://github.com/hyoseokp/single-helicity-coverage-bound}; the raw full-wave field data are
available from the author upon reasonable request.

\makeatletter\immediate\write\@auxout{\string\citation{apsrev42Control}}\makeatother
\bibliographystyle{apsrev4-2}
\bibliography{references}

\end{document}